\documentclass[3p,times,hidelinks]{article}
\usepackage[latin1]{inputenc}
\usepackage{bigfoot}
\usepackage{subfigure}
\usepackage{stmaryrd}
\usepackage{cite}
\usepackage{amssymb}
\usepackage{amsthm}
\usepackage{amsmath}

\usepackage[small, labelfont=bf, labelsep=period]{caption}
\usepackage{array}
\usepackage[shortlabels]{enumitem}
\usepackage{graphicx}
\usepackage{epsfig}
\usepackage{epstopdf}
\usepackage{fancyhdr}

\setlist[enumerate]{itemsep=0mm}

\theoremstyle{plain}% Theorem-like structures
\newtheorem{thm}{Theorem}[section]

\newtheorem{lem}[thm]{Lemma}

\newtheorem{algor}[thm]{Algorithm}

\theoremstyle{remark}

\usepackage[latin1]{inputenc}

\usepackage{epsfig}
\usepackage{epstopdf}
\usepackage{longtable}
\usepackage{framed}
\usepackage{setspace}
\usepackage{titlesec}
\usepackage{stmaryrd}
\titleformat{\section}{\large\bfseries}{\thesection}{1em}{}
\usepackage[margin=1.0in]{geometry}

\usepackage[shortlabels]{enumitem}
\setlist[enumerate]{itemsep=0mm}
\date{}
\author{Dohan Kim\footnote{E-mail: dkim@airesearch.kr}\\\\
\footnotesize A.I. Research Co., 2537-1 Kyungwon Plaza 201, Sinheung-dong, Sujeong-gu, \\\footnotesize Seongnam-si, Kyunggi-do, 461-811, South Korea\\}
\title{Distributed agent-based automated theorem proving in order-sorted first-order logic}

\begin{document}

\maketitle

\begin{abstract}
This paper presents a distributed agent-based automated theorem proving framework based on order-sorted first-order logic. Each agent in our framework has its own knowledge base, communicating to its neighboring agent(s) using message-passing algorithms. The communication language between agents is restricted in such a manner that each agent can only communicate to its neighboring agent(s) by means of their common language. In this paper we provide a refutation-complete report procedure for automated theorem proving in order-sorted first-order logic in a subclass of distributed agent-based networks. Rather than studying and evaluating the performance improvement of the automated theorem proving in order-sorted first-order logic using parallel or distributed agents, this paper focuses on building proofs in order-sorted first-order logic in a distributed manner under the restriction that agents may report their knowledge or observations only with their predefined language.
\end{abstract}

\section{Introduction}
\label{Introduction}
Automated theorem proving~\cite{fitting1996,loveland1978} is concerned with theorem proving using a computer program in an automatic manner and has been researched for several decades. Distributed (agent-based) automated theorem proving~\cite{conry1990,macintosh1991,fisher1997}, which is a subfield of automated theorem proving, allows each agent to start with a subset of the initial theory and to prove a target theorem in a coordinated and collaborative manner. It provides an agent with the capability that it may reason beyond its local knowledge~\cite{conry1990}. To prove a target theorem, each agent may perform its reasoning task automatically and concurrently, and reports knowledge to other agents if necessary~\cite{fisher1997}. One of the key issues in distributed automated theorem proving is the manner of communication between agents. Although each agent may or may not have a fixed language~\cite{keisler2012a,amir2005}, our approach is to restrict a pair of agents to communicate by means of their common language. This restricted communication can be employed for agents with report facilities, where each agent has limited privilege and is only allowed to report its knowledge or observations to its neighboring agent(s) with its predefined language.\\
\indent First-Order Logic (FOL) plays a key role in the knowledge representation for distributed automated theorem proving in that it often has a necessary expressive power for knowledge bases~\cite{amir2005,bourgne2011}. Furthermore, there exists a sound and refutation-complete resolution procedure for FOL~\cite{robinson1965}. (A resolution procedure is called \emph{refutation-complete} if it can derive a contradiction from every unsatisfiable set of formulae~\cite{inoue1991}.) However, in ordinary (unsorted) FOL, the universe of discourse is a single (unstructured) homogeneous set, showing a limited capability in terms of expressing sorted or structured information in a natural manner~\cite{cohn1989,walther1990}.\\
\indent Many-sorted FOL~\cite{abadi2010,cohn1989,Palacz2016} augments FOL by adding a set of sorts in its language, dividing the universe of discourse into (possibly overlapped) subsets called sorts. Ordinary (unsorted) FOL can therefore be viewed as one-sorted logic. The salient feature of many-sorted FOL is that it often increases deductive efficiency by means of the possibly smaller and divided search space along with the shorter deduction, avoiding some pointless branches of the search space~\cite{cohn1987,chien1998}. Knowledge representation and reasoning for many-sorted FOL have already been discussed in~\cite{walther1985,walther1990,cohn1989}. Note that there are many kinds of many-sorted FOLs. For instance, some many-sorted FOLs restrict each sort to be pairwise disjoint~\cite{abadi2010,enderton2001}, while others allow sorts to be overlapping~\cite{cohn1989,walther1990}. There are also approaches to using a partially-ordered (or preordered) set of sorts in the language of many-sorted FOL. A many-sorted FOL that has a partially-ordered (or preordered) set of sorts is specifically called an Order-Sorted FOL (OSFOL)\footnote{In this paper we use ``order-sorted first-order logic" and ``order-sorted logic'' interchangeably.}~\cite{beierle1992,frisch1991,kaneiwa2004,nelson2010,oberschelp1990,weibel1997}.\\
\indent In this paper we present a distributed automated theorem proving framework, in which agents and their communications are represented by a distributed agent graph in a distributed environment. Each agent has its own knowledge base containing OSFOL clauses, communicating to its neighboring agent(s) by means of their common language. \\
\indent The remainder of this paper is organized as follows. Section~\ref{section:preliminaries} gives a brief overview of OSFOL and its resolution procedure used in this paper. We also describe a signature tree in this section. Section~\ref{section:agentTP} presents our distributed automated theorem proving framework with OSFOL. In particular, we provide a refutation-complete OSFOL report procedure for automated theorem proving in a signature tree. Finally, we conclude in Section~\ref{section:conclusions}.

\section{Preliminaries}
\label{section:preliminaries}
In this section we summarize the necessary syntax and semantics of Order-Sorted First-Order Logic (OSFOL) along with its resolution procedure used in this paper. We also describe a signature network and knowledge base on a distributed agent graph in this section. The definitions and results in Sections~\ref{sec:syntax},~\ref{sec:semantics}, and~\ref{sec:resolution} are found in~\cite{beierle1992,cohn1989,frisch1991,ganesh2004,enderton2001,goguen1992,hedtstuck1990,kaneiwa2004,nelson2010,walther1988,weibel1997}. We assume the reader has some familiarity with FOL and its resolution procedure.

\subsection{Syntax}
\label{sec:syntax}
In contrast to the standard FOL, the language of OSFOL has \emph{restricted variables} to denote that each variable is restricted to range over a specific sort (i.e. a subset of the domain). A restricted variable is denoted in the form of $x\mbox{:}s$, where $x$ is a variable name and $s$ is a sort, respectively. We first describe a signature of OSFOL used in this paper.

An order-sorted signature $\Sigma=(S, P, F)$ consists of:
\begin{enumerate}[(1)]
\item A finite partially-ordered set of sorts $(S, \preceq)$, called the \emph{sort hierarchy}, with a greatest element $\top$ and a least element $\bot$.
\item An $S^*$-indexed family  $(P_w)_{w \in S^*}$ of sets of predicate symbols.
\item An $(S^* \times S)$-indexed family $(F_{w,s})_{w \in S^*,\,s \in S}$ of sets of function symbols. 
\end{enumerate}

For each sort $s \in S-\{\top, \bot\}$, there is a unary predicate $s(\cdots)$, called a \emph{sort predicate}. (For simplicity, a sort and its sort predicate are denoted by the same symbol in this paper. The distinction is clear from context.) The argument sort of a sort predicate $s(\cdots)$ is $\top$, i.e. $s \in P_{\top}$. We assume that there is at least one constant or ground term for every sort except $\bot$ to avoid problems involving \emph{empty sorts} discussed in~\cite{goguen1987,weidenbach1991}.

Given an order-sorted signature $\Sigma=(S, P, F)$, the set of variables of all sorts over $\Sigma$ is denoted by $V=\bigcup_{s \in S}V_s$, where $V_s$ denotes the set of variables for some $s \in S$. The sort of term $t$, denoted by $[t]$, is $s \in S$ if and only if $t \in V_s$ or $t=f(\cdots)$ and $f \in F_{w, s}$ for some $w \in S^*$. 
%If $t=f(\cdots)$, we simply denote the sort of term $t$ by $[f]$.

The set $T(\Sigma, V) = \bigcup_{s \in S}T_s(\Sigma, V)$ over $\Sigma$ and $V$, called the \emph{set of well-sorted terms} (or \emph{set of} $\Sigma$-\emph{terms} for short), is the smallest set of terms of all sorts satisfying:
\begin{enumerate}[(1)]
\item If $x\mbox{:}s \in V_s$, then $x\mbox{:}s \in T_s(\Sigma, V)$.
\item If $c \in F_{\varepsilon, s}$, then $c \in T_s(\Sigma, V)$. ($\varepsilon$ is the empty string of sorts.)
\item If $t_1, \ldots, t_j \in T(\Sigma, V), f \in F_{s_1\cdots s_j, s}$, and $[t_i] \preceq s_i$ for each $i$ with $1 \leq i \leq j$, then $f(t_1,\ldots, t_j) \in T_s(\Sigma, V)$.
\end{enumerate}

The \emph{set of well-sorted formulae} $\Phi (\Sigma, V)$ (or \emph{set of } $\Sigma$-\emph{formulae} for short) is defined inductively by:
\begin{enumerate}[(1)]
\item An atomic formula $p(t_1,\ldots, t_j)$ is in $\Phi (\Sigma, V)$ if $p \in P_{s_1\cdots s_j}$ with  $[t_i] \preceq s_i$ for each $i$ with $1\leq i\leq j$.
\item If $\phi, \psi \in \Phi(\Sigma, V)$, then so are $\neg \phi$, $\phi \vee \psi$, $\phi \wedge \psi$, and $\phi \Rightarrow \psi$.
\item If $x\mbox{:}s \in V_s$ and $\phi \in \Phi (\Sigma, V)$, then so are $\forall x\mbox{:}s\,.\,\phi$ and  $\exists x\mbox{:}s\,.\,\phi$
\end{enumerate}

Since the argument sort of each sort predicate is $\top$, atomic formulae involving sort predicates are always well-sorted. A $\Sigma$-formula is called a $\Sigma$-\emph{sentence} if the $\Sigma$-formula is closed, namely it does not contain any free variable. A \emph{well-sorted literal} (or $\Sigma$-\emph{literal} for short) is either a well-sorted atomic formula or its negation. A \emph{well-sorted clause} (or $\Sigma$-\emph{clause} for short) is defined as a finite disjunction of $\Sigma$-literals. It is also defined as a finite set of $\Sigma$-literals, which is synonymous with the universal closure of the disjunction of those $\Sigma$-literals. The empty clause is also a $\Sigma$-clause and is written as {\scriptsize$\square$}. A $\Sigma$-formula is said to be in Conjunctive Normal Form (CNF) if it is a conjunction of $\Sigma$-clauses, in which a conjunction of $\Sigma$-clauses can simply be denoted by the set of those $\Sigma$-clauses.

\subsection{Semantics}
\label{sec:semantics}
As shown in the previous section, an order-sorted signature contains a partially-ordered set of sorts. By applying a set-theoretic semantics, a sort is naturally interpreted as a subset of a given universe. The top sort $\top$ and the bottom sort $\bot$ are interpreted as the non-empty universe and the empty set, respectively.

Given an order-sorted signature $\Sigma=(S, P, F)$, a $\Sigma$-\emph{structure} $M$ is a pair $(A, I)$ such that: 
\begin{itemize}
\item $A=\{A_s\,|\, s \in S\}$ is an $S$-indexed family of sets.
\item $I$ is a function, called an \emph{interpretation}, where
\begin{enumerate}[(1)]
\item if $s \in S$, then $I(s)=A_s$. Furthermore, if $s_1 \preceq s_2$ for $s_1, s_2 \in S$, then $I(s_1) \subseteq I(s_2)$.
\item if $p \in P_w$ for $w=s_1\cdots s_n$, then $I(p) \subseteq I(s_1)\times\cdots\times I(s_n)$.
\item if $f \in F_{w, s}$ for $w=s_1\cdots s_n$, then $I(f):I(s_1)\times\cdots\times I(s_n) \rightarrow I(s)$.
\item if $s$ is a sort predicate symbol\footnote{A structure $M$ may interpret a sort predicate symbol in other ways~\cite{beierle1992,kaneiwa2004}, which is beyond the scope of this paper.} for some $s \in S$, then $I(s) = A_s$.
\end{enumerate}
\end{itemize}

A variable assignment $\alpha$ on a $\Sigma$-structure $M = (A, I)$ is an $S$-indexed family of functions $\alpha = \{\alpha_s:V_s \rightarrow A_s\,|\,s\in S\}$. For $x\mbox{:}s\in V_s$,  $\alpha_s(x\mbox{:}s)$ is simply denoted by $\alpha(x\mbox{:}s)$. The denotation $\llbracket t \rrbracket _\alpha$ of term $t$ under a variable assignment $\alpha$ is defined inductively by $\llbracket x\mbox{:}s \rrbracket _\alpha = \alpha(x\mbox{:}s)$ and $\llbracket f(t_1,\ldots, t_n)\rrbracket _\alpha = I(f)(\llbracket t_1 \rrbracket_\alpha,\ldots,\llbracket t_n \rrbracket_\alpha)$. 

We say that a $\Sigma$-structure $M$ \emph{satisfies} a $\Sigma$-formula $\phi$ with $\alpha$, denoted by $M \models_\Sigma \phi[\alpha]$, if the following holds:
\begin{enumerate}[(1)]
\item $M \models_\Sigma p(t_1,\ldots,t_n)[\alpha]$ iff $(\llbracket t_1 \rrbracket _\alpha,\ldots, \llbracket t_n \rrbracket _\alpha) \in I(p)$.
\item $M \models_\Sigma \neg\psi[\alpha]$ iff $M \not\models_\Sigma \psi[\alpha]$.
\item $M \models_\Sigma (\phi_1 \Rightarrow \phi_2)[\alpha]$ iff $M \models_\Sigma \neg\phi_1[\alpha]$ or $M \models_\Sigma\phi_2[\alpha]$.
\item $M \models_\Sigma \forall x\mbox{:}s\,.\,\psi[\alpha]$ iff for every $d\in A_s$, $M \models_\Sigma \psi[\alpha(x\mbox{:}s\,|\,d)]$, where $\alpha(x\mbox{:}s\,|\,d)$ maps the variable $x\mbox{:}s$ to $d$ and every other variable $y \in V$ to $\alpha(y)$.
\item $M \models_\Sigma \exists x\mbox{:}s\,.\,\psi[\alpha]$ iff there is some $d\in A_s$, $M \models_\Sigma \psi[\alpha(x\mbox{:}s\,|\,d)]$, where $\alpha(x\mbox{:}s\,|\,d)$ is as above.
\end{enumerate}

We say that a $\Sigma$-structure $M$ satisfies a $\Sigma$-formula $\phi$, denoted by $M \models_\Sigma \phi$, if $M$ satisfies $\phi$ with every variable assignment $\alpha$. Given a set $\Gamma$ of $\Sigma$-formulae, a $\Sigma$-structure $M$ is called a \emph{model} of $\Gamma$ if  for every $\Sigma$-formula $\psi \in \Gamma$, $M \models_\Sigma \psi$. Meanwhile, given a set $\Lambda$ of $\Sigma$-formulae, we say that $\Lambda$ is $\Sigma$-\emph{unsatisfiable} if there is no model of $\Lambda$ with respect to $\Sigma$.

Given a $\Sigma$-formula $\phi$, the \emph{relativization} $\widehat{\phi}$ of $\phi$ is the unsorted counterpart of $\phi$, which provides a means to an alternative semantics for OSFOL. For instance, $\phi_1\equiv \forall x\mbox{:}s\,.\,\psi$ and $\phi_2\equiv\exists x\mbox{:}s\,.\,\psi$ are relativized to $\widehat{\phi_1}\equiv\forall x\,.\,s(x) \Rightarrow \psi^\prime$ and $\widehat{\phi_2}\equiv\exists x\,.\,s(x) \wedge \psi^\prime$, respectively, where $\psi^\prime$ is the formula obtained from $\psi$ by substituting $x$ for all free occurrences of $x\mbox{:}s$ in $\psi$. Meanwhile, $c\in F_{\varepsilon,s}$ and $f \in F_{w, s}$ for $w=s_1\cdots s_n$ in an order-sorted signature $\Sigma$ are relativized to $s(c)$ and $\forall x_1,\ldots,x_n\,.\,(s_1(x_1) \wedge \cdots \wedge s_n(x_n))\Rightarrow s(f(x_1,\ldots,x_n))$, respectively. Let $\widehat{\Sigma}$ be the relativized version of $\Sigma$. Now, the connection between OSFOL and (unsorted) FOL in terms of semantics is described as the following theorem.

\begin{thm}\cite{beierle1992,weibel1997}
\label{thm:refcomp}
Let $\Sigma$ be an order-sorted signature and $\phi$ be a $\Sigma$-sentence. Then, $\phi$ is $\Sigma$-unsatisfiable iff $\widehat{\phi} \cup \widehat{\Sigma}$ is unsatisfiable.
\end{thm}

\subsection{Resolution procedure}
\label{sec:resolution}
A resolution procedure for FOL was first introduced by Robinson~\cite{robinson1965,wos1973}, which generally requires the transformation of FOL sentences into Conjunctive Normal Form (CNF)~\cite{socher1991}. The CNF transformation for an FOL sentence includes a  \emph{Skolemization}~\cite{stickel1988,mccune1988,frisch1991} procedure, where an FOL sentence $\phi$ is unsatisfiable if and only if its Skolemized sentence $\phi^\prime$ is unsatisfiable. We now briefly discuss the Skolemization procedure for a $\Sigma$-sentence.

The Skolemization procedure~\cite{mccune1988, frisch1991} for a $\Sigma$-sentence in prenex form is similar to that of FOL, where a $\Sigma$-sentence can be transformed into its equivalent prenex form. (The transformation process of a $\Sigma$-sentence into its prenex form and into CNF using the Skolemization procedure are described in~\cite{frisch1986,frisch1991}.) Let $y\mbox{:}s$ be an existentially quantified variable and let $x_1\mbox{:}s_1,\ldots,x_n\mbox{:}s_n$ be universally quantified variables such that $\exists y\mbox{:}s$ occurs in the scope of $x_1\mbox{:}s_1,\ldots,x_n\mbox{:}s_n$. Then, the existential quantifier $\exists y\mbox{:}s$ is removed, and all occurrences of the existentially quantified variable $y\mbox{:}s$ are replaced by $f(x_1\mbox{:}s_1,\ldots,x_n\mbox{:}s_n) \wedge s(f(x_1\mbox{:}s_1,\ldots,x_n\mbox{:}s_n))$, where $f$ is a new $n$-ary function symbol of sort $s$. If no universal quantifier is in the scope of $\exists y\mbox{:}s$, then $f(x_1\mbox{:}s_1,\ldots,x_n\mbox{:}s_n)$ is simply a Skolem constant $f$. For instance, $\forall x_1\mbox{:}s_1\,\exists y_1\mbox{:}s_2\,.\,\\ E(x_1\mbox{:}s_1\;y_1\mbox{:}s_2)$ is Skolemized to $\forall x_1\mbox{:}s_1\,.\,(E(x_1\mbox{:}s_1\;f(x_1\mbox{:}s_1)) \wedge s_2(f(x_1\mbox{:}s_1)))$. The following lemma says that the $\Sigma$-unsatisfiability is preserved in the Skolemization procedure for a $\Sigma$-sentence.

\begin{lem}\cite{frisch1991}
\label{lem:skolem}
A $\Sigma$-sentence $\phi$ is $\Sigma$-unsatisfiable iff its Skolemized $\Sigma$-sentence $\phi^\prime$ is $\Sigma$-unsatisfiable.
\end{lem}

The following definitions and the description of the $\Sigma$-resolution procedure are found in~\cite{frisch1991,weibel1997, beierle1992,Hustadt1999,kaneiwa2004,martelli1982,slagle1970}.

Let \emph{substitution} $\sigma$ be denoted by $\{x_1/t_1,\ldots,x_j/t_j\}$, where the terms $t_i$ are substituted for the variables $x_i$ for $1 \leq i \leq j$ and no $x_i$ occurs in any of $t_k$ for $1 \leq k \leq j$. For every $x \in V$, a substitution $\sigma$ is called \emph{well-sorted substitution} (or $\Sigma$-\emph{substitution} for short) if $\sigma x \in T(\Sigma, V)$ and $[\sigma x] \preceq [x]$. A $\Sigma$-substitution is extended to a mapping from $T(\Sigma, V)$ to $T(\Sigma, V)$. With some abuse of notation a $\Sigma$-substitution $\sigma$ is applicable to a $\Sigma$-formula (respectively, a set of $\Sigma$-formulae) by applying $\sigma$ to every $\Sigma$-term occurring in the $\Sigma$-formula (respectively, the set of $\Sigma$-formulae). Let $F=\{\phi_1,\ldots, \phi_n\}$ be a set of $\Sigma$-formulae and $\theta$ be a substitution. Then, $\theta$ is called a \emph{well-sorted unifier} (or $\Sigma$-\emph{unifier} for short) of $F$ with respect to $\Sigma$ if it is a $\Sigma$-substitution and is a unifier of $F$ (i.e. $\theta(\phi_1)=\cdots = \theta(\phi_n)$). If there is a $\Sigma$-unifier of $F$, then $F$ is said to be \emph{well-sorted unifiable} (or $\Sigma$-\emph{unifiable} for short).

Similarly to the resolution procedure for FOL, the resolution procedure for OSFOL, called the \emph{order-sorted resolution procedure} (or $\Sigma$-\emph{resolution procedure} for short) operates on $\Sigma$-clauses and is based on unification. We assume that for any set of sorts $G\subset S$, the greatest lower bound of $G$ exists in a sort hierarchy $(S, \preceq)$, which assures the existence of a unique $\Sigma$-\emph{most general unifier} (or $\Sigma$-\emph{mgu} for short) for any $\Sigma$-unifiable set of $\Sigma$-formulae~\cite{beierle1992,walther1988}. If it is not the case, synthetic greatest lower bounds on the sorts can be added to a sort hierarchy in order to ensure the existence of a unitary $\Sigma$-unifier for a given $\Sigma$-unifiable set of $\Sigma$-formulae. The interested reader may refer to~\cite{weibel1997} for further details. 

Let $p_1$ and $p_2$ be predicate symbols in an order-sorted signature. If two atomic formulae $p_1(r_1,\ldots,r_j)$ and $p_2(s_1,\ldots,s_j)$ are identical, then there is nothing to unify. Meanwhile, if predicate symbols $p_1$ and $p_2$ are different, they cannot be unified at all. Otherwise, if $p_1 = p_2$ and two atomic formulae $p_1(r_1,\ldots,r_j)$ and $p_2(s_1,\ldots,s_j)$ are not identical, then consider a set $X= \{r_i \stackrel{\text{\tiny ?}}{=} s_i\,|\, i=1,\ldots,j\}$ of temporary equations on $\Sigma$-terms, which is extracted from $p_1(r_1,\ldots,r_j)$ and $p_2(s_1,\ldots,s_j)$. We use a binary predicate symbol $\stackrel{\text{\tiny ?}}{=}$ in $X$ to denote temporary equations on $T(\Sigma, V)$. If $\Sigma$-substitution $\tau$ makes all pairs of $\Sigma$-terms $r_i$ and $s_i$ syntactically equivalent, namely $\tau r_i = \tau s_i$ for $i=1,\ldots,j$, then $\tau$ unifies $p_1(r_1,\ldots,r_j)$ and $p_2(s_1,\ldots,s_j)$. 

The order-sorted unification algorithm with $\Sigma$-mgu performs any of the following routine until only (permanent) equations remain. If the algorithm terminates without failure and the remaining equations are $x_1\mbox{:}s_1=t_1,\ldots, x_n\mbox{:}s_n=t_n$, then it yields a substitution $\sigma$ (i.e. $\Sigma$-mgu) with $\{x_1\mbox{:}s_1/t_1,\ldots,x_n\mbox{:}s_n/t_n\}$.
\begin{enumerate}[(1)]
\item Select any temporary equation in $X$ having the form $x\mbox{:}s \stackrel{\text{\tiny ?}}{=} x\mbox{:}s$, where $x\mbox{:}s$ is a variable. Then, remove it from $X$.
\item Select any temporary equation in $X$ having the form $t \stackrel{\text{\tiny ?}}{=} x\mbox{:}s$, where $x\mbox{:}s$ is a variable and $t$ is a non-variable $\Sigma$-term. Then, replace it with $x\mbox{:}s \stackrel{\text{\tiny ?}}{=} t$ in $X$.
\item  Select any temporary equation in $X$ having the form $f_1(t^\prime_1,\ldots,t^\prime_i) \stackrel{\text{\tiny ?}}{=} f_2(t^{\prime\prime}_1,\ldots,t^{\prime\prime}_j)$, where $f_1$ and $f_2$ are function symbols. (A constant symbol is considered as a 0-ary function symbol here.) If $f_1\neq f_2$ or $i \neq j$, then return failure. Otherwise, replace it with the temporary equations $t^\prime_1\stackrel{\text{\tiny ?}}{=} t^{\prime\prime}_1,\ldots,t^\prime_i \stackrel{\text{\tiny ?}}{=} t^{\prime\prime}_i$ in $X$.
\item  Select any temporary equation in $X$ having the form $y\mbox{:}s^\prime\stackrel{\text{\tiny ?}}{=}t$, where $y\mbox{:}s^\prime$ is a variable of sort $s^\prime$ and $t$ is a non-variable $\Sigma$-term of sort $s^{\prime\prime}$. If $y\mbox{:}s^\prime$ occurs in $t$, or if $s^{\prime\prime} \not\preceq s^\prime$, then return failure. Otherwise, apply the substitution $\{y\mbox{:}s^\prime / t\}$ to all other temporary equations and replace $y\mbox{:}s^\prime\stackrel{\text{\tiny ?}}{=}t$ with the equation $y\mbox{:}s^\prime = t$ in $X$.
\item Select any unmarked temporary equation in $X$ having the form $y\mbox{:}s^\prime\stackrel{\text{\tiny ?}}{=}z\mbox{:}s^{\prime\prime}$, where $y\mbox{:}s^\prime$ and $z\mbox{:}s^{\prime\prime}$ are distinct variables. If $s^{\prime\prime} \preceq s^\prime$, then apply the substitution $\{y\mbox{:}s^\prime / z\mbox{:}s^{\prime\prime}\}$ to all other temporary equations and replace $y\mbox{:}s^\prime\stackrel{\text{\tiny ?}}{=}z\mbox{:}s^{\prime\prime}$ with the equation $y\mbox{:}s^\prime = z\mbox{:}s^{\prime\prime}$ in $X$. Otherwise, if $s^\prime \prec s^{\prime\prime}$, then apply the substitution $\{z\mbox{:}s^{\prime\prime} / y\mbox{:}s^\prime\}$ to all other temporary equations and replace $y\mbox{:}s^\prime\stackrel{\text{\tiny ?}}{=}z\mbox{:}s^{\prime\prime}$ with the equation $z\mbox{:}s^{\prime\prime} = y\mbox{:}s^\prime$ in $X$. Now, consider the case where $s^\prime \not\preceq s^{\prime\prime}$ and $s^{\prime\prime} \not\prec s^\prime$. If the greatest common subsort of $s^\prime$ and $s^{\prime\prime}$ is $\bot$, then return failure. Otherwise, let $x\mbox{:}s$ be a new variable of sort $s$, where $s$ is the greatest common subsort of $s^\prime$ and $s^{\prime\prime}$. Then, apply the substitution $\{y\mbox{:}s^\prime / x\mbox{:}s, \;z\mbox{:}s^{\prime\prime} / x\mbox{:}s\}$ to all other temporary equations and replace $y\mbox{:}s^\prime\stackrel{\text{\tiny ?}}{=}z\mbox{:}s^{\prime\prime}$ with the equations $y\mbox{:}s^\prime = x\mbox{:}s, \;z\mbox{:}s^{\prime\prime} = x\mbox{:}s$ in $X$.
\end{enumerate}

Let $C_1$ and $C_2$ be $\Sigma$-clauses, in which variables in $C_1$ and variables in $C_2$ are standardized apart. For some $m \subseteq C_1$ and $n \subseteq C_2$, if $m \cup \neg n$ is $\Sigma$-unifiable by $\Sigma$-mgu $\sigma$, then $(C_1-m)\sigma \cup (C_2-n)\sigma$ is a $\Sigma$-\emph{resolvent} of $C_1$ and $C_2$. Let $C$ be a $\Sigma$-clause and $L_1, L_2$ be $\Sigma$-literals such that $L_1$ and $L_2$ are $\Sigma$-unifiable by $\Sigma$-mgu $\sigma$. If $C=\{L_1, L_2\}\cup D$, then ($\{L_1\}\cup D)\sigma$ is called a \emph{factor} of $C$. The \emph{factoring} is the associated operation that derives ($\{L_1\}\cup D)\sigma$ from $C$. For instance, $\{P(x\mbox{:}s)\}$ is a factor of $\{P(x\mbox{:}s), P(y\mbox{:}s)\}$. 

Let $\Gamma$ be a set of $\Sigma$-clauses. A sequence $C_1,\ldots, C_k$ of one or more $\Sigma$-clauses is called a \emph{derivation} of $C_k$ from $\Gamma$ by $\Sigma$-\emph{resolution}, denoted by $\Gamma \vdash_{\Sigma\mbox{-}res} C_k$, if each $\Sigma$-clause in the sequence is one of the followings: (i) a $\Sigma$-clause in $\Gamma$, (ii) a $\Sigma$-resolvent of earlier $\Sigma$-clauses, (iii) a factor of an earlier $\Sigma$-clause in the sequence. Similarly, we say that $\Gamma$ yields a set $I$ of $\Sigma$-clauses by $\Sigma$-resolution, denoted by $\Gamma \vdash_{\Sigma\mbox{-}res} I$, if each $\Sigma$-clause in $I$ is derivable from $\Gamma$ by $\Sigma$-resolution.

We next describe a hybrid reasoning system~\cite{frisch1991,weidenbach1991} consisting of a sort module and OSFOL formulae, in which a sort module describes a sort hierarchy. A sort module is often described by a special language or using the first-order language~\cite{weidenbach1991,cohn1987}.  In this paper a sort module is described by FOL. To ensure that the $\Sigma$-resolution procedure is refutation-complete, it suffices that a sort module is of the \emph{definite program} form, in which the \emph{definite program} consists of \emph{definite clauses}~\cite{weibel1997,frisch1991}. (Recall that a definite clause is a clause that has exactly one positive literal~\cite{inoue2004}.)
In what follows we assume that a sort module for any set of $\Sigma$-clauses is of the definite program form, which is sufficient to express an essential sort hierarchy as discussed in~\cite{weibel1997}. 

\begin{thm}\cite{frisch1991,weibel1997}
\label{thm:refcomplete}
A set $\Gamma$ of $\Sigma$-clauses is $\Sigma$-unsatisfiable iff $\Gamma \vdash_{\Sigma\mbox{-}res}$ \scriptsize $\square$  .
\end{thm}

\begin{thm}\cite{frisch1991,weibel1997}
\label{thm:Herbrand}
Let $A$ be a set of $\Sigma$-clauses. Then, $A$ is $\Sigma$-unsatisfiable iff there is a finite $\Sigma$-unsatisfiable set $A^\prime$ of $\Sigma$-ground clauses of $A$.
\end{thm}

Theorem~\ref{thm:refcomplete} describes the refutation-completeness of the $\Sigma$-resolution procedure. Meanwhile, Theorem~\ref{thm:Herbrand}, called the \emph{Sorted Herbrand Theorem}, says that the Herbrand Theorem~\cite{robinson1965} can be extended to a set of $\Sigma$-clauses.

We say that $\Sigma$-clause $C$ \emph{subsumes} $\Sigma$-clause $C^\prime$, denoted by $C \geq_{\Sigma} C^\prime$, if there is a $\Sigma$-substitution $\theta$ such that $C\theta \subseteq C^\prime$~\cite{plotkin1970,frisch1991}. Note that the empty clause subsumes every $\Sigma$-clause. A set $A$ of $\Sigma$-clauses subsumes a set $B$ of $\Sigma$-clauses, denoted by $A \geq_{\Sigma} B$, if every $\Sigma$-clause in $B$ is subsumed by some $\Sigma$-clause in $A$~\cite{slagle1970}. The \emph{Lifting Theorem} for $\Sigma$-resolution is shown in the following theorem. 

\begin{thm}\cite{frisch1991}
\label{thm:Lifting}
Let $\Gamma$ be a set of $\Sigma$-clauses and $\Gamma^\prime$ be the set of the corresponding $\Sigma$-ground clauses of $\Gamma$. Let $C^\prime$ be a $\Sigma$-ground clause such that $\Gamma^\prime \vdash_{\Sigma\mbox{-}res} C^\prime$. Then, there is a $\Sigma$-clause $C \geq_{\Sigma}C^\prime$ such that  $\Gamma \vdash_{\Sigma\mbox{-}res} C$.
\end{thm}
%\vfill
\subsection{A signature network on a distributed agent graph}

We use the following definitions involving a signature network and knowledge base found in~\cite{keisler2012a,keisler2012b,mcilraith2001}.

Let $G=(V, E)$ be a directed graph that consists of a finite set $V$ of vertices and a set $E \subseteq V\times V$ of edges. A directed graph $G=(V, E)$ has a \emph{decider} $D \in V$ if there exists a path from every other vertex $x \in V$ to $D$. A directed graph $G=(V, E)$ is called a \emph{pointed graph} if it has at least one decider. 

Let $V$ denote a set of agents and $E$ denote their communication links. A \emph{signature network} on $G=(V, E)$ is defined as an object $\mathbb{S}=(V, E, L(\cdot))$, in which $G$ is a pointed graph and $L(\cdot)$ is a labeling that assigns a signature $L(a)$ to each agent $a\in V$. We let $\mathcal{L}(a)$ denote the language built with $L(a)$ and call the set $\mathcal{L}(a)$ the \emph{language of agent} $a$ for $a\in V$. Given a signature network $\mathbb{S} = (V, E, L(\cdot))$, a \emph{knowledge base} (or \emph{theory}) \emph{over} $\mathbb{S}$ is defined as an object $\mathbb{K}=(V, E, L(\cdot), K(\cdot))$, in which $K(\cdot)$ is a labeling that assigns a knowledge base $K(a) \subset \mathcal{L}(a)$ to each agent $a \in V$. The \emph{combined signature} $L(V)$ and \emph{combined knowledge base} $K(V)$ are defined as $L(V)=\bigcup_{a\in V}L(a)$ and $K(V)=\bigcup_{a \in V}K(a)$, respectively. For any set $\Omega$ of non-logical symbols, a signature network $\mathbb{S}=(V, E, L(\cdot))$ is said to have the \emph{peak property} if the subgraph comprising those agents in $V$ whose signature contains $\Omega$ has a decider.

A signature network on a directed tree $G=(V, E)$ is called a \emph{signature tree} if it has the peak property and for every vertex $x$, which is not a decider, there is a unique edge $(x, y) \in E$. Therefore, a signature tree has a unique decider $D$ such that for every vertex $x \neq D$, there is a unique path from $x$ to $D$.

We define a \emph{distributed agent graph} $G_d=(V_d, E_d)$ as a directed acyclic graph of distributed agents, in which each vertex $v \in V_d$ denotes a distributed agent and $(u, v) \in E_d$ denotes that agent $u$ reports its knowledge or observations to agent $v$. Each vertex of $V_d$ has a unique label denoting its identifier in $G_d$. We say that agent $v$ is an \emph{immediate successor} of agent $u$, and $u$ is an \emph{immediate predecessor} of agent $v$ if $(u, v) \in E_d$, A \emph{source agent} is an agent that has no immediate predecessor. We assume throughout that $G_d=(V_d, E_d)$ is a distributed agent graph.

In this paper we consider a signature network and knowledge base on $G_d=(V_d, E_d)$. For instance, $\mathbb{S}=(V_d, E_d, L(\cdot))$ denotes a signature network on $G_d$ and $\mathbb{K}=(V_d, E_d, L(\cdot), K(\cdot))$ denotes a knowledge base over $\mathbb{S}=(V_d, E_d, L(\cdot))$. In the remainder of this paper $P(u)$ denotes the set of predicate symbols in $L(u)$ for $u \in V_d$, $P(C)$ denotes the set of predicate symbols in clause $C$, and $P(u, v)$ denotes the set $P(u) \cap P(v)$ for $u, v \in V_d$. Similarly, $L(C)$ denotes the set of non-logical symbols in clause $C$, $l(u, v)$ denotes the set $L(u) \cap L(v)$ for $u, v \in V_d$, and $\mathcal{L}(l(u, v))$ denotes the corresponding language built with $l(u, v)$. For a $\Sigma$-formula $w$, $L(w)$ denotes the set of non-logical symbols in $w$. We assume throughout that $P(u, v) \neq \emptyset$ for each edge $(u, v) \in E_d$ in a signature network $\mathbb{S}=(V_d, E_d, L(\cdot))$ and that each knowledge base is consistent before the query is given to a decider. 
%We also assume that $|V_d| \geq 2$ in a distributed agent graph $G_d=(V_d, E_d)$ unless otherwise stated.
%\vfill
\section{The OSFOL report procedure for automated theorem proving}
\label{section:agentTP}
In this section we present our distributed agent-based automated theorem proving framework based on OSFOL, where each distributed agent in a network reports knowledge or observations to its neighboring agent(s) in order to build proofs using message-passing algorithms. Each agent in our framework has its own knowledge base, communicating to its neighboring agent(s) by means of their common language in a distributed environment.

\begin{algor}\normalfont OSFOL-SEND$(u, v, K^\prime(u))$\\
\setlist{nosep}
\noindent Input: $u$, $v$ for $(u, v) \in E_d$ in $G_d=(V_d, E_d)$, and a set $K^\prime(u)$ of $\Sigma$-clauses in $\mathcal{L}(u)$.
\begin{itemize}[leftmargin=*]
\item Create sets $U$ and $W$ if they do not already exist. Initialize them to $\emptyset$.
\item For each $\Sigma$-clause $C \in K^\prime(u)$,
\begin{itemize}
\item If $P(C) \subset P(u, v)$, then set $X \leftarrow L(C)\setminus l(u, v)$. ($X$ is the set of function or constant symbols in $L(C)$ that cannot be sent directly to agent $v$.)
\item If $X = \emptyset$ or $C$ is the empty clause, then add $C$ to a set $U$. Otherwise, add $C$ to a set $W$.
\end{itemize}
\item If $W$ is non-empty, then un-Skolemize $W$ into $W^\prime$. If successful and $L(w) \subset l(u, v)$ for every un-Skolemized $\Sigma$-formula $w \in W^\prime$, then add each element of $W^\prime$ to $U$. Otherwise, return failure.
\item Send $U$ to agent $v$.
\end{itemize}
\label{Algorithm:rmsending}
\end{algor}

Algorithm~\ref{Algorithm:rmsending}, which is based on the FOL resolution-based message-sending procedure discussed in~\cite{amir2005}, describes the OSFOL message-sending procedure from agent $u$ to agent $v$ for $(u, v) \in E_d$ in a distributed agent graph $G_d=(V_d, E_d)$. When a $\Sigma$-clause $\phi \notin \mathcal{L}(l(u, v))$ for $(u, v) \in E_d$, agent $u$ cannot send $\phi$ to agent $v$ directly because it is not in their common language. In this case an un-Skolemized $\Sigma$-formula $\phi^\prime$ of $\phi$ can be sent from agent $u$ to agent $v$ if $P(\phi)\subset P(u, v)$ and $\phi^\prime \in \mathcal{L}(l(u, v))$. Specifically, if $P(\phi)\subset P(u, v)$, it is desirable to ensure that $L(\phi^\prime) \subset l(u, v)$. To the best of our knowledge, there is no known un-Skolemization procedure for a $\Sigma$-clause set. We use our un-Skolemization procedure for a $\Sigma$-clause set in Algorithm~\ref{Algorithm:rmsending}, which is mostly based on the McCune's un-Skolemization procedure~\cite{mccune1988} for an FOL clause set. Although the McCune's un-Skolemization procedure for an FOL clause set is sound~\cite{mccune1988}, not every clause set can be un-Skolemized. Similary, not every $\Sigma$-clause set can be un-Skolemized in our un-Skolemization procedure. Therefore, we restrict the input of our un-Skolemization procedure and assume throughout that functions used in Algorithm~\ref{Algorithm:rmsending} are \emph{acceptable} for our un-Skolemization procedure, where every Skolem function is naturally acceptable to our un-Skolemization procedure.  We now briefly discusses our un-Skolemization procedure for a $\Sigma$-clause set, which is based on the McCune's un-Skolemization procedure~\cite{mccune1988} for an FOL clause set. We first describe the un-Skolemization procedure for a single $\Sigma$-clause. The sufficient conditions for un-Skolemizing a single $\Sigma$-clause in this paper are as follows:

\begin{enumerate}[(i)]
\item For each function symbol considered by un-Skolemization, its (variable) arguments are all distinct and it does not contain any non-variable term argument.
\item Let $f$ be an $m$-ary function symbol and $g$ be an $n$-ary function symbol such that $m \leq n$. Then, the set of  (variable) arguments of $f$ are contained in the set of (variable) arguments of $g$.
\item No two functions headed by the identical function symbol appear together in a single $\Sigma$-clause.
\end{enumerate}

We say that functions that satisfy the above conditions are \emph{acceptable} for our un-Skolemization procedure. Now the un-Skolemization procedure for a single $\Sigma$-clause is the reverse of the Skolemization procedure for a single $\Sigma$-clause discussed in Section~\ref{sec:resolution}. For instance, let $D \in P_{s_1s_2s_3}$ and consider symbols $f, g$, and $h$ for un-Skolemization. We see that $D(f, g(x\mbox{:}s_1), h(x\mbox{:}s_1, y\mbox{:}s_2)) \wedge s_1(f) \wedge s_2(g(x\mbox{:}s_1)) \wedge s_3(h(x\mbox{:}s_1, y\mbox{:}s_2))$ is un-Skolemized to $\exists v_1\mbox{:}s_1 \forall x\mbox{:}s_1 \exists v_2\mbox{:}s_2 \forall y\mbox{:}s_2 \exists v_3\mbox{:}s_3\,.\,D(v_1, v_2, v_3)$.

The un-Skolemization procedure for a $\Sigma$-clause set is basically the same with the un-
Skolemization procedure for an FOL clause set~\cite{mccune1988} except the consideration of sorts. %The reader is encouraged to read the McCune's un-Skolemization procedure~\cite{mccune1988} for details. 
The following steps summarize the un-Skolemization procedure for a $\Sigma$-clause set. We call each function (respectively, function symbol) considered by un-Skolemization as a \emph{Skolem expression} (respectively, \emph{Skolem symbol}). We assume that each $\Sigma$-clause in a $\Sigma$-clause set satisfies the above conditions (i)--(iii). Therefore, if two Skolem expressions have a common Skolem symbol, they are originated in two different $\Sigma$-clauses.
\begin{enumerate}[(1)]
\item Maximally partition a $\Sigma$-clause set in such a manner that no two partitions share a Skolem symbol. Then, for each partition, perform the following steps (2)--(6).
\item Rename variables in such a manner that two variable sets from any pair of $\Sigma$-clauses are disjoint.
\item For each identical function symbol found in the set of Skolem expressions, unify the set of Skolem expressions headed by that function symbol such that only one Skolem expression remains for each function symbol and that the (variable) arguments of the resulting Skolem expression are all distinct. If successful, the unifying substitution is applied to the entire partition. Otherwise, return failure. Note that the unifying substitution here is only a renaming of variables of the same sort rather than $\Sigma$-substitution.
\item Make every $n$-ary Skolem expression have the same (variable) arguments by unification. The order of its (variable) arguments is irrelevant here. Furthermore, if $f$ is an $n$-ary Skolem symbol and $g$ is an $m$-ary Skolem symbol such that $n \leq m$, force the set of (variable) arguments of $f$ to be contained in the set of (variable) arguments of $g$. If successful, the unifying substitution is applied to the entire partition. Otherwise, return failure.
\item For all Skolem expressions in the partition, construct the quantifier prefix and replace Skolem expressions with the corresponding existentially quantified (sorted) variables.
\item Add the resulting un-Skolemized $\Sigma$-formulae to the set $W$. 
\item If every partition can be un-Skolemized by steps (2)--(6), the resulting set $W$ is interpreted as a conjunction of those un-Skolemized $\Sigma$-formulae. 
\end{enumerate}

We next give an example to illustrate the above steps. Sort predicates that appear in the un-Skolemization procedure for a single $\Sigma$-clause are omitted because it is clear from the corresponding predicate symbols in the signature. Now consider the following partition of three $\Sigma$-clauses. For $p \in P_{s_1s_2s_3s_4}$, $q \in P_{s_1s_2s_3s_4}$, and $r \in P_{s_1s_2s_3s_4}$,

\indent 1. $p(x_1\mbox{:}s_1, x_2\mbox{:}s_2, f_1(x_1\mbox{:}s_1), g_1(x_2\mbox{:}s_2, x_1\mbox{:}s_1))$,\\
\indent 2. $q(y_1\mbox{:}s_1, y_2\mbox{:}s_2, f_2(y_1\mbox{:}s_1), g_1(y_2\mbox{:}s_2, y_1\mbox{:}s_1))$,\\
\indent 3. $r(z_1\mbox{:}s_1, z_2\mbox{:}s_2, f_2(z_1\mbox{:}s_1), g_2(z_1\mbox{:}s_1, z_2\mbox{:}s_2))$.\\

\noindent After applying the step (3) procedure with substitution $\{y_1\mbox{:}s_1/x_1\mbox{:}s_1, y_2\mbox{:}s_2/x_2\mbox{:}s_2\}$ and  $\{z_1\mbox{:}s_1/x_1\mbox{:}s_1\}$ for $g_1$ and $f_2$, respectively, the above $\Sigma$-clauses become as follows:

\indent 1. $p(x_1\mbox{:}s_1, x_2\mbox{:}s_2, f_1(x_1\mbox{:}s_1), g_1(x_2\mbox{:}s_2, x_1\mbox{:}s_1))$,\\
\indent 2. $q(x_1\mbox{:}s_1, x_2\mbox{:}s_2, f_2(x_1\mbox{:}s_1), g_1(x_2\mbox{:}s_2, x_1\mbox{:}s_1))$,\\
\indent 3. $r(x_1\mbox{:}s_1, z_2\mbox{:}s_2, f_2(x_1\mbox{:}s_1), g_2(x_1\mbox{:}s_1, z_2\mbox{:}s_2))$.\\

\noindent After applying the step (4) procedure with substitution $\{z_2\mbox{:}s_2/x_2\mbox{:}s_2\}$, the above $\Sigma$-clauses become as follows:

\indent 1. $p(x_1\mbox{:}s_1, x_2\mbox{:}s_2, f_1(x_1\mbox{:}s_1), g_1(x_2\mbox{:}s_2, x_1\mbox{:}s_1))$,\\
\indent 2. $q(x_1\mbox{:}s_1, x_2\mbox{:}s_2, f_2(x_1\mbox{:}s_1), g_1(x_2\mbox{:}s_2, x_1\mbox{:}s_1))$,\\
\indent 3. $r(x_1\mbox{:}s_1, x_2\mbox{:}s_2, f_2(x_1\mbox{:}s_1), g_2(x_1\mbox{:}s_1, x_2\mbox{:}s_2))$.\\

\noindent After applying the steps (5) and (6) procedure, the set $W$ of un-Skolemized $\Sigma$-formulae becomes as follows: \\
$\{\text{Q\,.\,} p(x_1\mbox{:}s_1, x_2\mbox{:}s_2, v_1\mbox{:}s_3, v_3\mbox{:}s_4), \text{Q\,.\,} q(x_1\mbox{:}s_1, x_2\mbox{:}s_2, v_2\mbox{:}s_3, v_3\mbox{:}s_4), \text{Q\,.\,} r(x_1\mbox{:}s_1, x_2\mbox{:}s_2, v_2\mbox{:}s_3, v_4\mbox{:}s_4)\}$, where the quantifier prefix $\text{Q}=\forall x_1\mbox{:}s_1\exists v_1\mbox{:}s_3\exists v_2\mbox{:}s_3\forall x_2\mbox{:}s_2 \exists v_3\mbox{:}s_4\exists v_4\mbox{:}s_4$. The following lemma says that the un-Skolemization procedure for a $\Sigma$-clause set is sound.

\begin{lem}
\label{lem:unskolemization}
If the un-Skolemization procedure for a $\Sigma$-clause set succeeds and yields a set $W$ of un-Skolemized $\Sigma$-formulae, then $W$ is $\Sigma$-unsatisfiable iff the original $\Sigma$-clause set is $\Sigma$-unsatisfiable.
\end{lem}
%\begin{rem}
\noindent \emph{Remarks.} McCune presented the un-Skolemization procedure for an FOL clause set and showed that it is sound~\cite{mccune1988}. The proof of Lemma~\ref{lem:unskolemization} is produced by transforming the McCune's proof~\cite{mccune1988} involving an FOL clause set into one involving a $\Sigma$-clause set. Note that our un-Skolemization procedure for a $\Sigma$-clause set does not involve any equality used in the McCune's un-Skolemized procedure by using the restricted form of Skolem expressions (see (i)--(iii)).

\begin{proof}
Assume that the procedure succeeds and yields a set of un-Skolemized $\Sigma$-formulae. We show that steps (3)--(4) preserve logical equivalence and step (5) preserves the $\Sigma$-unsatisfiability. It is easy to see that steps (1)--(2) and steps (6)--(7) preserve logical equivalence.

Unification procedures in steps (3)--(4) rename variables of the same sort if successful and do not attempt to unify two different (sorted) variables in the same $\Sigma$-clause (see (i)--(iii)). Thus, steps (3)--(4) preserve logical equivalence. Since a set of $\Sigma$-clauses at the start of step (5) is a Skolemization of a resulting set of un-Skolemized $\Sigma$-formulae yielded by step (5), the $\Sigma$-unsatisfiability is preserved at step (5) by Lemma~\ref{lem:skolem}.
\end{proof}

The following algorithm describes the OSFOL message-receiving procedure for agent $v$, which processes a set of the received $\Sigma$-formulae from agent $u$ for $(u, v)\in E_d$ in $G_d=(V_d, E_d)$ by using the Skolemization procedure.

\begin{algor}\normalfont OSFOL-RECV$(U, u, v, K^\prime(v))$\\
\setlist{nosep}
\noindent Input: $u$, $v$ for $(u, v) \in E_d$ in $G_d=(V_d, E_d)$, a set $K^\prime(v)$ of $\Sigma$-clauses in $\mathcal{L}(v)$, and a set $U$ of $\Sigma$-formulae received from agent $u$.
\begin{itemize}[leftmargin=*]
\item Skolemize a set $U$ of $\Sigma$-formulae into $U^\prime$.
\item For each $\Sigma$-clause $C \in U^\prime$, add $C$ to $K^\prime(v)$.
\end{itemize}
\label{Algorithm:rmreceiving}
\end{algor}

The OSFOL resolution-based report procedure incorporates Algorithm~\ref{Algorithm:rmsending} and~\ref{Algorithm:rmreceiving} to prove a query in CNF using the $\Sigma$-resolution procedure. In Algorithm~\ref{Algorithm:rmreport} each $K^\prime(a)$ for $a \in V_d$ is composed of knowledge base $K(a)$ and its associated temporary knowledge base to save $\Sigma$-resolvents, etc. If a decider $D$ proves query $Q$ in CNF by Algorithm~\ref{Algorithm:rmreport}, it adds $Q$ to $K(D)$. Then, each temporary knowledge base built during an automated theorem proving procedure is removed. In what follows we assume that each query and its negation are given under CNF. 

\begin{algor}\normalfont OSFOL-REPORT$(G_d, (K(i))_{i\in V_d}, D, Q)$\\
\setlist{nosep}
\noindent Input: A distributed agent graph $G_d=(V_d, E_d)$, the collection of knowledge bases $(K(i))_{i\in V_d}$, a decider $D \in V_d$, and query $Q \in \mathcal{L}(D)$.
\begin{itemize}[leftmargin=*]
\item For each agent $a \in V_d$, construct $K^\prime(a)$ from $K(a)$. Add $\neg Q$ to $K^\prime(D)$.
\item Let $d(a, D)$ be the corresponding distance function from agent $a$ in $V_d$ to the decider agent $D$. Find some agent $u$ such that $d(u, D)$ is maximum. 
\item Concurrently, for every $(u, v) \in E_d$ such that $d(u, D)>d(v, D)$,
\begin{itemize}
\item Agent $u$:\footnote{If $|V_d|=1$, a distributed agent graph has a unique agent, which is the decider. In this case run this subroutine directly with $u=D$ after adding $\neg Q$ to $K^\prime(D)$ that has been constructed from $K(D)$.},\\
\noindent Perform the $\Sigma$-resolution procedure and add $\Sigma$-resolvents to $K^\prime(u)$. \\
\noindent Case $u = D$: If the empty clause can be derived, return success. Otherwise, return failure.\\
\noindent Case $u \neq D$: Call OSFOL-SEND($u$, $v$, $K^\prime(u)$), where $(u, v) \in E_d$.
\item Agent $v$:\\
\noindent When a set $U$ of $\Sigma$-formulae arrives from agent $u$, call OSFOL-RECV$(U, \\u, v, K^\prime(v))$. Once the receiving procedure has been completed, set $u:=v$ and continue the loop.
\end{itemize}
\end{itemize}
\label{Algorithm:rmreport}
\end{algor}

Now consider what happens when Algorithm~\ref{Algorithm:rmreport} runs on a signature tree $\mathbb{S}=(V_d, E_d, L(\cdot))$ for $|V_d| \geq 2$. Each source agent performs the $\Sigma$-resolution procedure and then calls the OSFOL-SEND procedure in order to send $\Sigma$-clauses including $\Sigma$-resolvents to its unique immediate successor. Note that it does not call the OSFOL-RECV procedure at all. Meanwhile, the decider agent performs the OSFOL-RECV procedure and the $\Sigma$-resolution procedure, but does not call the OSFOL-SEND procedure. Other kinds of agents perform the OSFOL-RECV procedure, the $\Sigma$-resolution procedure, and the OSFOL-SEND procedure when running Algorithm~\ref{Algorithm:rmreport} on the signature tree.\\
\indent In our approach  different report facilities can be assigned to a group of agents by restricting the language of each agent and its communications in a well-defined manner. Each agent is not allowed to report its knowledge or observations beyond its language. This approach has in common with a \emph{syslog}~\cite{singer2004} system logger in a UNIX environment in that different facilities (i.e. kernel, ftp, mail, etc.) are handled differently based on their configurations. However, the reporting capability of each agent can be predefined at a language level rather than a system-specific configuration level in our approach.

\begin{figure}[ht] 
{\small
$\bf{Sort\;module\;(FOL\;representation)\mbox{:} }$ $F(\Sigma)=\{$\\$ 
\forall x\,.\,W(x) \rightarrow A(x),\;\;\;\;\; \forall x\,.\, F(x) \rightarrow A(x),\;\;\;\;\; \forall x\,.\, B(x) \rightarrow A(x),$\\$ 
\forall x\,.\,C(x) \rightarrow A(x),\;\;\;\;\;\; \forall x\,.\, S(x) \rightarrow A(x),\;\;\;\;\;\; \forall x\,.\, G(x) \rightarrow P(x),$\\$
W(w),\;\;\; F(f),\;\;\; B(b),\;\;\; C(c),\;\;\; S(s),\;\;\; G(g)\}$\\$\newline
\textbf{Input\;(OSFOL\;representation)\mbox{:}} $\\$
(1)\;E(a_1\mbox{:}A\;p_1\mbox{:}P)\vee \bar{M}(a_2\mbox{:}A\;a_1\mbox{:}A)\vee \bar{E}(a_2\mbox{:}A\;p_2\mbox{:}P)\vee E(a_1\mbox{:}A\;a_2\mbox{:}A)$\\$
(2)\;M(c_1\mbox{:}C\;b_1\mbox{:}B)\;\;\; (3)\;M(s_1\mbox{:}S\;b_1\mbox{:}B)\;\;(4)\;M(b_1\mbox{:}B\;f_1\mbox{:}F)$\\$
(5)\;M(f_1\mbox{:}F\;w_1\mbox{:}W)\;\;(6)\;\bar{E}(w_1\mbox{:}W\;f_1\mbox{:}F)\;\;(7)\;\bar{E}(w_1\mbox{:}W\;g_1\mbox{:}G)$\\$
(8)\;E(b_1\mbox{:}B\;c_1\mbox{:}C)\;\;\;(9)\;\bar{E}(b_1\mbox{:}B\;s_1\mbox{:}S)\;\;(10)\;P(h(c_1\mbox{:}C))$\\$
(11)\;E(c_1\mbox{:}C\;h(c_1\mbox{:}C))\;\;\;(12)\;P(i(s_1\mbox{:}S))\;\;(13)\;G(j(a_1\mbox{:}A\;a_2\mbox{:}A))$\\$(14)\;E(s_1\mbox{:}S\;i(s_1\mbox{:}S))$\\$\newline
\textbf{Negation\;of\;query\;Q\mbox{:}}$\\
$\indent (\neg Q)\mbox{:}\;\bar{E}(a_1\mbox{:}A\;a_2\mbox{:}A)\vee \bar{E}(a_2\mbox{:}A\;j(a_1\mbox{:}A\;a_2\mbox{:}A))$
}
\caption{Schubert's Steamroller problem~\cite{walther1985}.}
\label{fig:Schubert}
\end{figure}

\indent We next show how the OSFOL report procedure can be applied to the classical \emph{Schubert's Steamroller Problem}~\cite{walther1985,stickel1986}, which is well-studied topic in many-sorted FOL. It is naturally fit into the OSFOL setting, since it involves the partially-ordered set of sorts. The Schubert's Steamroller Problem is found in~\cite{walther1985} and is given as follows:\\\\
``Wolves, foxes, birds, caterpillars, and snails are animals, and there are some of each of them. Also there are some grains, and grains are plants. Every animal either likes to eat all plants or all animals
much smaller than itself that like to eat some plants. Caterpillars and snails are much smaller than birds, which are much smaller than foxes, which in turn are much smaller than wolves. Wolves do
not like to eat foxes or grains, while birds like to eat caterpillars but not snails. Caterpillars and snails like to eat some plants. Therefore there is an animal that likes to eat a grain-eating animal.''\\\\
In~\cite{walther1985} the following predicates are used for the Schubert's Steamroller Problem:\\\\
\begin{tabular}{ll}
$A(t)$: $t$ is an animal,& $W(t)$: $t$ is a wolf,\\
$F(t)$: $t$ is a fox,&$B(t)$: $t$ is a bird,\\
$C(t)$: $t$ is a caterpillar,& $S(t)$: $t$ is a snail,\\
$G(t)$: $t$ is a grain,& $P(t)$: $t$ is a plant,\\
$M(st)$: $s$ is much smaller than $t$,& $E(st)$: $s$ likes to eat $t$.
\end{tabular}
\\\\
\indent Figure~\ref{fig:Schubert} describes an OSFOL representation of Schubert's Steamroller problem in clause notation. We use the Frisch's hybrid model~\cite{frisch1991}, where the sort module is represented by the standard FOL. The third line of the sort module in Figure~\ref{fig:Schubert} indicates that sorts $W, F, B, C, S$, and $G$ are not empty. Note that the sort module, which describes the sort hierarchy, is only used for $\Sigma$-substitutions. The sort module in Figure~\ref{fig:Schubert} shows that sorts $W, F, B, C, S$ are subsorts of sort $A$, while sort $G$ is the subsort of sort $P$. Function symbols $h$, $i$, and $j$ in Figure~\ref{fig:Schubert} are Skolem symbols employed for the Skolemization procedure. Note that query $Q$ is negated in Figure~\ref{fig:Schubert} to find if the empty clause can be derived from the input in Figure~\ref{fig:Schubert} using the $\Sigma$-resolution procedure. We now consider distributed agents each of which has its signature and knowledge base. 

\begin{figure}[h!]
\small
{
$\;\;\;\;\;\;\;\;\;\;\;\;\;\;\;\;\;\;\;\;\neg Q=\bar{E}(a_1\mbox{:}A\;a_2\mbox{:}A)\vee\bar{E}(a_2\mbox{:}A\;j(a_1\mbox{:}A\;a_2\mbox{:}A))$\\
$\text{}\;\;\;\;\;\;\;\;\;\;\;\;\;\;\;\;\;\;\;\;L(x)=\{E \in P_{A\top},\;M \in P_{AA},\;j\in F_{AA, G}\}$\\
$\text{}\;\;\;\;\;\;\;\;\;\;\;\;\;\;\;\;\;\;\;\;K(x)=\{E(a_1\mbox{:}A\;p_1\mbox{:}P)\vee \bar{M}(a_2\mbox{:}A\;a_1\mbox{:}A)\vee\bar{E}(a_2\mbox{:}A\;p_2\mbox{:}P)\vee E(a_1\mbox{:}A\;a_2\mbox{:}A), G(j(a_1\mbox{:}A\;a_2\mbox{:}A))\}\cup F(\Sigma)$
%}
\begin{center}
\includegraphics[width=0.25\textwidth]{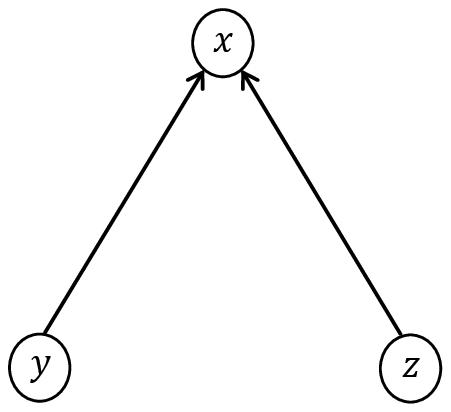}
\end{center}
%\footnotesize{
\begin{tabular}{ll}
\;\;\;\;\;\;\;\;\;\;\;\;\;\;\;\;\;\;\;\;\;\;\;\;\;\;\;\;$L(y)=\{E \in P_{A\top},\;h \in F_{C, P},\;i \in F_{S, P}\}$ & $\;\;L(z)=\{M \in P_{AA}\}$\\
\;\;\;\;\;\;\;\;\;\;\;\;\;\;\;\;\;\;\;\;\;\;\;\;\;\;\;\;$K(y)=\{\bar{E}(w_1\mbox{:}W\;f_1\mbox{:}F),\bar{E}(w_1\mbox{:}W\;g_1\mbox{:}G),$&$\;\;K(z)=\{M(c_1\mbox{:}C\;b_1\mbox{:}B), M(s_1\mbox{:}S\;b_1\mbox{:}B),$\\ 
\;\;\;\;\;\;\;\;\;\;\;\;\;\;\;\;\;\;\;\;\;\;\;\;\;\;\;\;$E(b_1\mbox{:}B\;c_1\mbox{:}C),\bar{E}(b_1\mbox{:}B\;s_1\mbox{:}S),$&$\;\;M(b_1\mbox{:}B\;f_1\mbox{:}F), M(f_1\mbox{:}F\;w_1\mbox{:}W)\}\cup F(\Sigma)$\\ 
\;\;\;\;\;\;\;\;\;\;\;\;\;\;\;\;\;\;\;\;\;\;\;\;\;\;\;\;$E(c_1\mbox{:}C\;h(c_1\mbox{:}C)), P(h(c_1\mbox{:}C)),$&$$\\ 
\;\;\;\;\;\;\;\;\;\;\;\;\;\;\;\;\;\;\;\;\;\;\;\;\;\;\;\;$E(s_1\mbox{:}S\;i(s_1\mbox{:}S)), P(i(s_1\mbox{:}S))\}\cup F(\Sigma)$ &
\end{tabular}
}
\caption{A distributed knowledge base $\mathbb{K}=(V_d, E_d, L(\cdot), K(\cdot))$ for Figure~\ref{fig:Schubert}.}
\label{fig:dagents2}
\end{figure}

\indent Figure~\ref{fig:dagents2} shows a distributed knowledge base $\mathbb{K}=(V_d, E_d, L(\cdot), K(\cdot))$ for $|V_d|=3$ with respect to the combined knowledge base shown in Figure~\ref{fig:Schubert}. To simplify the notation, we also denote a signature\footnote{We assume that each agent in a signature network is equipped with the same static built-in sort module, so each agent does not need to report the sort module to each other. Therefore, we omit the sort hierarchy in each agent's signature and assume that it is implicitly included in each agent's signature.} $L(a)$ as a set of non-logical symbols assigned to agent $a\in V_d$. Agent $y$ and $z$ report a set of $\Sigma$-formulae to agent $x$ by the OSFOL-SEND procedure in Algorithm~\ref{Algorithm:rmreport}. We see that the common predicate symbol between agents $x$ and $y$ is $E \in P_{A\top}$. Agent $y$ can neither report the $\Sigma$-clause $E(c_1\mbox{:}C\,h(c_1\mbox{:}C))$ nor $E(s_1\mbox{:}S\,i(s_1\mbox{:}S))$ to agent x directly because $h, i \notin l(x, y)$. Since functions $h$ and $i$ are acceptable for our un-Skolemization procedure (see Appendix A), agent $y$ can report the un-Skolemized $\Sigma$-formulae $\forall c_1\mbox{:}C\exists p_1\mbox{:}P\,.\,E(c_1\mbox{:}C\,p_1\mbox{:}P)$ and $\forall s_1\mbox{:}S\exists p_2\mbox{:}P\,.\, E(s_1\mbox{:}S\, p_2\mbox{:}P)$ to agent $x$. Note that these $\Sigma$-formulae are not $\Sigma$-clauses, although the Skolemization of them are $\Sigma$-clauses. The reported $\Sigma$-formulae from agent $y$ to agent $x$ by the OSFOL-SEND procedure in Algorithm~\ref{Algorithm:rmreport} are as follows:\\

\noindent
{\small
(1)\;$\bar{E}(w_1\mbox{:}W f_1\mbox{:}F)\;\;\;\;\;\; \;\;\;\;\;\;\;\;\;\,\;\;\;\;\,\;\;(2)\;\bar{E}(w_1\mbox{:}W g_1\mbox{:}G)$\;\;\;\;\;\; \;\;\;\;\;\;\;\;\;\,\;\;\;\;\, (3)\;$E(b_1\mbox{:}B\;c_1\mbox{:}C)$\\ (4)\;$\bar{E}(b_1\mbox{:}B\;s_1\mbox{:}S)$\;\;\;\;\;\; \;\;\;\;\;\;\;\;\;\;\,\;\;\;\;\, (5)\;$\forall c_1\mbox{:}C\exists p_1\mbox{:}P\,.\,E(c_1\mbox{:}C\;p_1\mbox{:}P)$\\ (6)\;$\forall s_1\mbox{:}S\exists p_2\mbox{:}P\,.\, E(s_1\mbox{:}S\;p_2\mbox{:}P)$\\
}

In (5) and (6) the universal quantifiers are not omitted because the order of quantifiers has to be considered. By the OSFOL-RECV procedure in Algorithm~\ref{Algorithm:rmreport} agent $x$ receives those $\Sigma$-formulae, and then adds them to $K^\prime(x)$ after Skolemization. Meanwhile, the reported $\Sigma$-formulae from agent $z$ to agent $x$ by the OSFOL-SEND procedure in Algorithm~\ref{Algorithm:rmreport} are as follows:\\

\noindent
{\small
(7)\;$M(c_1\mbox{:}C\;b_1\mbox{:}B)$ \;\;\;\;\;\;\;\;\;\;\,\;\;\;\;\,\;\;\;\;\,(8)\;$M(s_1\mbox{:}S\;b_1\mbox{:}B)$\\ (9)\;$M(b_1\mbox{:}B\;f_1\mbox{:}F)$ \;\;\;\;\;\;\;\;\;\;\;\;\;\;\,\;\;\;\;\,(10)\;$M(f_1\mbox{:}F\;w_1\mbox{:}W)$\\\\
}
\indent Since the above $\Sigma$-formulae do not contain any existential quantifier, agent $x$ does not need to Skolemize the reported $\Sigma$-formulae from agent $z$. In contrast, (5) and (6) are Skolemized in agent $x$ by the OSFOL-RECV procedure in Algorithm~\ref{Algorithm:rmreport} using Skolem symbols $SK_1$ and $SK_2$ which have not been used in agent $x$. Since $SK_1$ and $SK_2$ are non-logical symbols, they are now added to $L(x)$. After the receiving procedures from agent $y$ and agent $z$ including Skolemization have been completed, $K^\prime(x)$ contains the following $\Sigma$-clauses:\\

\noindent
{\small
(1)\;$\bar{E}(w_1\mbox{:}W\;f_1\mbox{:}F)$ \;\;\;\;\;\;\;\;\;\;\;\;\,\;\;\;\;\,(2)\;$\bar{E}(w_1\mbox{:}W\;g_1\mbox{:}G)$\\ (3)\;$E(b_1\mbox{:}B\;c_1\mbox{:}C)$ \;\;\;\;\;\;\;\;\;\;\;\;\;\;\;\;\;\;\,\;(4)\;$\bar{E}(b_1\mbox{:}B\;s_1\mbox{:}S)$\\ (5)\;$E(c_1\mbox{:}C\;SK_1(c_1\mbox{:}C))$\;\;\;\;\;\;\,\;\;\;\;\,(6)\;$P(SK_1(c_1\mbox{:}C))$\\
(7)\;$E(s_1\mbox{:}S\;SK_2(s_1\mbox{:}S))$ \;\;\;\;\;\;\;\;\;\;\,(8)\;$P(SK_2(s_1\mbox{:}S))$\\ (9)\;$M(c_1\mbox{:}C\;b_1\mbox{:}B)$ \;\;\;\;\;\;\;\;\;\;\;\;\;\;\;\;\;\,(10)\;$M(s_1\mbox{:}S\;b_1\mbox{:}B)$\\ (11)\;$M(b_1\mbox{:}B\;f_1\mbox{:}F)$\;\;\;\;\;\;\;\;\;\;\; \;\;\;\;\,(12)\;$M(f_1\mbox{:}F\;w_1\mbox{:}W)$\\
(13)\;$E(a_1\mbox{:}A\;p_1\mbox{:}P)\vee \bar{M}(a_2\mbox{:}A\;a_1\mbox{:}A)\vee \bar{E}(a_2\mbox{:}A\;p_2\mbox{:}P)\vee E(a_1\mbox{:}A\;a_2\mbox{:}A)$\\
(14)\;$G(j(a_1\mbox{:}A\;a_2\mbox{:}A))$\\
(15)\;$\bar{E}(a_1\mbox{:}A\;a_2\mbox{:}A)\vee \bar{E}(a_2\mbox{:}A\;j(a_1\mbox{:}A\;a_2\mbox{:}A))$\\\\
}
\indent Clauses (1)--(8) are those received from agent $y$, whereas clauses (9)--(12) are those received from agent $z$. Clauses (13)--(14) are the clauses that already exist in $K(x)$. Clause (15) is the negation of query $Q$ that has been added to $K^\prime(x)$ by Algorithm~\ref{Algorithm:rmreport}. The remaining steps for decider $x$ are to use the $\Sigma$-resolution procedure and to find if the empty clause can be derived from $K^\prime(x)$. Since the Schubert's Steamroller Problem has already been solved using many-sorted logic with an improved deductive efficiency than that of FOL~\cite{walther1985, weibel1997}, we use the similar steps found in~\cite{walther1985}:\\

\noindent
{\small
(16)\;$E(a_1\mbox{:}A\;p_1\mbox{:}P)\vee \bar{M}(a_2\mbox{:}A\;a_1\mbox{:}A)\vee \bar{E}(a_2\mbox{:}A\;p_2\mbox{:}P)\vee\bar{E}(a_2\mbox{:}A\;j(a_1\mbox{:}A\;a_2\mbox{:}A))$; 13(4) + 15(1)\\
(17)\;$E(a_1\mbox{:}A\;p_1\mbox{:}P)\vee \bar{M}(a_2\mbox{:}A\;a_1\mbox{:}A)\vee\bar{E}(a_2\mbox{:}A\;j(a_1\mbox{:}A\;a_2\mbox{:}A))$; factoring from (16)\\
(18)\;$E(w_1\mbox{:}W\;p_1\mbox{:}P)\vee \bar{E}(f_1\mbox{:}F\;j(w_1\mbox{:}W\;f_1\mbox{:}F))$; 17(2) + 12(1)\\
(19)\;$E(f_1\mbox{:}F\;p_1\mbox{:}P)\vee \bar{E}(b_1\mbox{:}B\;j(f_1\mbox{:}F\;b_1\mbox{:}B))$; 17(2) + 11(1)\\
(20)\;$\bar{E}(f_1\mbox{:}F\;j(w_1\mbox{:}W\;f_1\mbox{:}F))$; 18(1) + 2(1)\\
(21)\;$\bar{E}(b_1\mbox{:}B\;j(f_1\mbox{:}F\;b_1\mbox{:}B))$; 19(1) + 20(1)\\
(22)\;$E(b_1\mbox{:}B\;p_1\mbox{:}P)\vee \bar{M}(s_1\mbox{:}S\;b_1\mbox{:}B)\vee \bar{E}(s_1\mbox{:}S\;p_2\mbox{:}P)$; 13(4) + 4(1)\\
(23)\;$\bar{M}(s_1\mbox{:}S\;b_1\mbox{:}B)\vee \bar{E}(s_1\mbox{:}S\;p_2\mbox{:}P)$; 21(1) + 22(1)\\
(24)\;$\bar{E}(s_1\mbox{:}S\;p_2\mbox{:}P)$; 23(1) + 10(1)\\
(25)\;$\square$\;; 24(1) + 7(1)\\\\
}
\indent Note that agent $y$ reports $\Sigma$-formulae involving predicate symbol $E$, whereas agent $z$ reports $\Sigma$-formulae involving predicate symbol $M$ to agent $x$. Given a query $Q$, agent $x$ collects reports from agents $y$ and $z$, finding if the empty clause can be derived from $K^\prime(x)$ by Algorithm~\ref{Algorithm:rmreport}. If so, the query $Q$ is then added as a theorem to $K(x)$.

We next discuss the main results of this paper. First, it is easy to see that the resolution rule for $\Sigma$-ground clauses is the same as the resolution rule for propositional clauses. The next lemma therefore follows directly from Theorem 2 in~\cite{slagle1970}. Recall that a set of $\Sigma$-clauses is synonymous with a $\Sigma$-formula that is a conjunction of all those $\Sigma$-clauses (see Section~\ref{sec:syntax}).

\begin{lem}
\label{lem:rescraig}
Let $A$ and $B$ be finite sets (conjunction) of $\Sigma$-ground clauses, where $B$ is not $\Sigma$-unsatisfiable. If $A\,\&\,B$ is $\Sigma$-unsatisfiable, then there is a finite set $I$ of $\Sigma$-ground clauses such that (a) $A \vdash_{\Sigma\mbox{-}res} I$, (b) $I$ subsumes any CNF of $\neg B$, and therefore (c) $I\,\&\, B$ is $\Sigma$-unsatisfiable.
\end{lem}

The following theorem is based on the \emph{Interpolation Theorem} for FOL resolution (see Theorem 3 in~\cite{slagle1970}). The main difference is that the former uses the Sorted Herbrand Theorem (see Theorem~\ref{thm:Herbrand}) instead of the Herbrand Theorem along with the Lifting Theorem for $\Sigma$-resolution (see Theorem~\ref{thm:Lifting}).

\begin{thm}
\label{thm:craig}
Let $A$ and $B$ be finite sets (conjunction) of $\Sigma$-clauses, where $B$ is not $\Sigma$-unsatisfiable. If $A\,\&\,B$ is $\Sigma$-unsatisfiable, then there is a finite set $I$ of $\Sigma$-clauses such that (a) $A \vdash_{\Sigma\mbox{-}res} I$, (b) $I\,\&\, B$ is $\Sigma$-unsatisfiable, and (c) every predicate symbol occurring in $I$ occurs in both $A$ and $B$.
\end{thm}
\begin{proof}
By Theorem~\ref{thm:Herbrand}, if $A\,\&\,B$ is $\Sigma$-unsatisfiable, then there are finitely many $\Sigma$-ground clauses $A_1,\ldots,A_j$ of $A$ and finitely many $\Sigma$-ground clauses $B_1,\ldots,B_k$ of $B$ such that $A_1\,\&\,\cdots\,\&\,A_j\,\&\,B_1\,\&\,\cdots\,\&\,B_k$ is $\Sigma$-unsatisfiable. Then, by Lemma~\ref{lem:rescraig}, there is a finite set $I_g$ of $\Sigma$-ground clauses such that (1) $A_1\,\&\,\cdots\, \&\,A_j \\\vdash_{\Sigma\mbox{-}res} I_g$, (2) $I_g$ subsumes any CNF of $\neg B_1\vee\ldots\vee \neg B_k$, and therefore (3) $I_g\,\&\, B_1\,\&\,\cdots\,\&\,B_k$ is $\Sigma$-unsatisfiable. Then, by Theorem~\ref{thm:Lifting}, we have $A \vdash_{\Sigma\mbox{-}res} I$ such that $I \geq_\Sigma I_g$. It follows that $I$ subsumes any CNF of $\neg B_1\vee\ldots\vee\neg B_k$, and therefore $I\&\,B_1\,\&\,\cdots\,\&\,B_k$ is $\Sigma$-unsatisfiable. By Theorem~\ref{thm:Herbrand}, we have that $I\,\&\, B$ is $\Sigma$-unsatisfiable. Since $I$ subsumes any CNF of $\neg B_1\vee\ldots\vee\neg B_k$ and $A \vdash_{\Sigma\mbox{-}res} I$, every predicate symbol occurring in $I$ occurs in both $A$ and $B$. 
\end{proof}

\begin{thm}
\label{thm:interpolation}
Let $D$ be a decider agent and $x \neq D$ be an agent in a signature tree $\mathbb{S}=(V_d, E_d, L(\cdot))$ such that $(x, D) \in E_d$. Let $Q \in \mathcal{L}(D)$ be a query and let $\bar{K}(D):=K(D)\,\&\,\neg Q$. If $K(x)\,\&\,\bar{K}(D)$ is $\Sigma$-unsatisfiable, the $\Sigma$-unsatisfiability is obtained at $D$ by Algorithm~\ref{Algorithm:rmreport}.
\end{thm}
\begin{proof}
If $\bar{K}(D)$ is $\Sigma$-unsatisfiable, the proof is trivial. Assume that $\bar{K}(D)$ is not $\Sigma$-unsatisfiable. By Theorem~\ref{thm:craig}, there is a finite set $I$ of $\Sigma$-clauses such that (a) $K(x) \vdash_{\Sigma\mbox{-}res} I$, (b) $I\,\&\, \bar{K}(D)$ is $\Sigma$-unsatisfiable, and (c) every predicate symbol occurring in $I$ occurs in both $K(x)$ and $\bar{K}(D)$. Let $u(I)$ be a resulting set of our un-Skolemization procedure applied to $I$. By part (c) of Theorem~\ref{thm:craig}, $P(C) \subset P(x, D)$ for each $\Sigma$-clause $C \in I$, which follows that the set $U$ of $\Sigma$-formulae that are sent from agent $x$ to agent $D$ includes $u(I)$ (up to variable renaming) by the OSFOL-SEND procedure in Algorithm~\ref{Algorithm:rmreport} and our assumption that functions used in the OSFOL-SEND procedure are acceptable for the un-Skolemization procedure. \\\indent For a given set $X$, let $sk(X)$ denote a Skolemized set of $X$, which is obtained by Skolemizing each element of $X$. Since $I\,\&\, \bar{K}(D)$ is $\Sigma$-unsatisfiable, $sk(u(I))\&\, \bar{K}(D)$ is $\Sigma$-unsatisfiable by Lemmas~\ref{lem:skolem} and~\ref{lem:unskolemization}. By the OSFOL-RECV procedure in Algorithm~\ref{Algorithm:rmreport}, the received set $U$ from agent $x$ is Skolemized to $sk(U)$ at agent $D$. Since $sk(u(I))\&\, \bar{K}(D)$ is $\Sigma$-unsatisfiable, $sk(U)\,\&\, \bar{K}(D)$ is $\Sigma$-unsatisfiable at $D$. Thus, the $\Sigma$-unsatisfiability is obtained at $D$ by Algorithm~\ref{Algorithm:rmreport}.
\end{proof}

The OSFOL report procedure can be viewed as a distributed resolution (theorem proving) procedure in that given a query $Q$, it performs a (refutation) theorem proving process using $\Sigma$-resolution rules in a distributed manner. Recall that a resolution procedure is refutation-complete if it can derive the empty clause from every unsatisfiable set of clauses. The following theorem says that our OSFOL report procedure for automated theorem proving is refutation-complete.

\begin{thm}
\label{thm:algo3e1}
Let $\mathbb{K}=(V_d, E_d, L(\cdot), K(\cdot))$ be a knowledge base over a signature tree $\mathbb{S}=(V_d, E_d, L(\cdot))$. Given a decider $D \in V_d$ and a query $Q \in \mathcal{L}(D)$, $K(V_d)\models_\Sigma Q$ iff Algorithm~\ref{Algorithm:rmreport} returns success.
\end{thm}
\begin{proof}
$(\Leftarrow)$\\
Assume Algorithm~\ref{Algorithm:rmreport} returns success. Then, $K(V_d)$$\,\&\,\neg Q$ is $\Sigma$-unsatisfiable by Theorem~\ref{thm:refcomplete}. Thus, $K(V_d)\models_\Sigma Q$.\\
$(\Rightarrow)$\\
Assume $K(V_d)\models_\Sigma Q$. Then,  $K(V_d)\,\&\,\neg Q$ is $\Sigma$-unsatisfiable. We proceed by induction on the number $|V_d|$ of agents. If there is only one agent, which is a decider, the result follows directly from Theorem~\ref{thm:refcomplete}. Now, assume the result holds for $|V_d|=n$ for a positive integer $n$ as an inductive hypothesis and consider the case where $|V_d|=n+1$. We show that the $\Sigma$-unsatisfiability is obtained at a decider $D$ by Algorithm~\ref{Algorithm:rmreport}, which follows that Algorithm~\ref{Algorithm:rmreport} returns success at a decider $D$ by Theorem~\ref{thm:refcomplete}.\\
\indent Since $\mathbb{S}$ is a signature tree, we can choose a source agent $x \in V_d$ such that $x \neq D$ and that there is a unique edge $(x, y) \in E_d$ for $y \in V_d$. Further, let $V^\prime_d = V_d \setminus \{x\}$ and $E^\prime_d = E_d \setminus \{(x, y)\}$. We have $L(x) \cap L(y) = L(x) \cap L(V_d^\prime)$ by the peak property. Since $K(V_d)\,\&\,\neg Q$ is $\Sigma$-unsatisfiable by assumption, $K(x)\,\&\,K(V_d^\prime)\,\&\,\neg Q$ is $\Sigma$-unsatisfiable. By letting $\bar{K}(V_d^\prime)=K(V_d^\prime)\,\&\,\neg Q$, $K(x)\,\&\,\bar{K}(V_d^\prime)$ is $\Sigma$-unsatisfiable. Let $U$ be a set of $\Sigma$-formulae that are sent from agent $x$ to agent $y$ by the OSFOL-SEND procedure in Algorithm~\ref{Algorithm:rmreport}. Then, $sk(U)\,\&\,\bar{K}(V_d^\prime)$ is $\Sigma$-unsatisfiable (see the proof of Theorem~\ref{thm:interpolation}). By letting $\bar{\bar{K}}(V_d^\prime):= sk(U)\,\&\,\bar{K}(V_d^\prime)$, the $\Sigma$-unsatisfiability is obtained at $V_d^\prime$. By the inductive hypothesis, the $\Sigma$-unsatisfiability is obtained at a decider $D$ by Algorithm~\ref{Algorithm:rmreport}.
\end{proof}

\section{Concluding remarks}
\label{section:conclusions}
This paper discussed a distributed agent-based automated theorem proving framework using the $\Sigma$-resolution procedure in order-sorted first-order logic. Each agent is only allowed to report its knowledge or observations to its neighboring agent(s) by means of their common language in a distributed agent-based environment. In other words, when building proofs in order-sorted first-order logic, agents in our framework are restricted to report their knowledge or observations only with their predefined language. Therefore, the language-level control of reports from distributed agents is allowed in our framework when building proofs in order-sorted first-order logic in a distributed manner. We also used the assumptions that a sort module expressed by first-order logic is of the definite program form and that functions used in the OSFOL report procedure are acceptable for our un-Skolemization procedure. With these assumptions we established the first refutation-complete report procedure, to the best of our knowledge, for automated theorem proving in order-sorted first-order logic on a signature tree.
%Our approach can also be used for distributed automated theorem proving in classical first-order logic by reducing a many-sorted universe of order-sorted first-order logic to a single-sorted homogeneous universe. 

\nocite{*}
\bibliographystyle{plain}
%{
\footnotesize
\begin{spacing}{0.1}
\bibliographystyle{plain}
\bibliography{dkim}
\end{spacing}

\end{document}